\documentclass[a4paper,USenglish,final]{article}
\usepackage{etex}
\usepackage[normalem]{ulem}
\usepackage{amsmath}
\usepackage{amssymb}
\usepackage{mathtools}
\usepackage[vlined,ruled,linesnumbered]{algorithm2e}
\usepackage[utf8]{inputenc} 
\usepackage[T1]{fontenc}
\usepackage{pgf}
\usepackage{tikz}
\usepackage[sort,comma,square,numbers]{natbib}
\usepackage[bookmarks=false, pdftex]{hyperref}
\usepackage{xspace}
\usepackage{paralist}
\usepackage{cutwin}
\usepackage{tabulary}
\usepackage{todonotes}
\usepackage[notcite,notref]{showkeys}
\usepackage{comment}
\usepackage{marginnote}
\usepackage{nicefrac}
\usepackage{xspace}

\title{Neighbor Joining And Leaf Status}

\usepackage{authblk}
\usepackage[margin=25mm]{geometry}
\usepackage{amsthm}
\author[1,2]{Mathias Weller}
\affil[1]{TU Berlin, Germany}
\affil[2]{CNRS, LIGM, Université Gustave Eiffel, Paris, France}

\usepackage[sort&compress, capitalize, nameinlink]{cleveref}

\usepackage{amsfonts}
\usepackage{ifthen}
\usepackage{paralist}
\usepackage{xspace}
\usepackage{txfonts}

\usepackage{xcolor}
\colorlet{darkgreen}{green!50!black}
\colorlet{dg}{darkgreen}
\colorlet{medgray}{gray!75}
\colorlet{lightgray}{gray!30}
\definecolor{linkcol}{rgb}{0,0,0.4} 
\definecolor{citecol}{rgb}{0.5,0,0} 

\usepackage{tikz}
\usetikzlibrary{arrows.meta,shapes,positioning,calc,decorations.pathreplacing,decorations.markings,decorations.pathmorphing,patterns}
\pgfdeclarelayer{background2}
\pgfdeclarelayer{background}
\pgfdeclarelayer{foreground}
\pgfsetlayers{background2,background,main,foreground}

\tikzstyle{bold}=[draw, line width=2pt]
\tikzstyle{optional}=[dashed]
\tikzstyle{path}=[decorate, decoration={snake, amplitude=.6mm}]

\tikzstyle{small}=[inner sep=2pt]
\tikzstyle{tiny}=[inner sep=1.7pt]
\tikzstyle{textnode}=[inner sep=0pt]

\tikzstyle{triangle}=[draw, regular polygon, regular polygon sides=3]
\tikzstyle{vertex}=[circle, draw, fill=white]
\tikzstyle{reti}=[vertex, fill=black]
\tikzstyle{leaf}=[vertex, rectangle]
\tikzstyle{leaf2}=[vertex, regular polygon, regular polygon sides=3]

\tikzstyle{smallvertex}=[vertex, small]
\tikzstyle{smallleaf}=[leaf, inner sep=3.3pt]
\tikzstyle{smallleaf2}=[leaf2, inner sep=1.7pt]
\tikzstyle{smalltriangle}=[triangle, inner sep=1.5pt]
\tikzstyle{smallreti}=[reti, small]

\tikzstyle{tinyvertex}=[vertex, tiny]

\tikzstyle{toroot}=[smallvertex, fill=white]
\tikzstyle{normal}=[smallvertex, fill=black]

\tikzstyle{match}=[edge,line width=3pt]
\tikzstyle{bgmatch}=[match, opacity=.5]

\tikzstyle{edge}=[draw,-]
\tikzstyle{matching}=[edge,line width=3pt]
\tikzstyle{solution}=[gray!60, line width=5pt]
\tikzstyle{arc}=[draw,arrows={-Latex[length=6pt]}]
\tikzstyle{boldarc}=[draw, bold, arrows={-Latex[length=10pt]}]
\tikzstyle{revarc}=[draw, arrows={Latex[length=6pt]-}]
\tikzstyle{boldrevarc}=[draw, bold, arrows={Latex[length=10pt]-}]
\tikzstyle{HRL}=[gray, bgmatch]

\ifx\newdefinition\undefined
  \let\newdefinition\newtheorem
\fi
\ifx\definition\undefined
  \newdefinition{definition}{Definition}
  \newtheorem{theorem}{Theorem}
  \newtheorem{lemma}{Lemma}
  \newtheorem{corollary}{Corollary}
\fi

\newtheorem{observation}{Observation}

\newdefinition{construction}{Construction}

\ifx\qedhere\undefined
  \newcommand\qedhere{\hfill\qed}
\fi

\newcommand{\w}{\ensuremath{\omega}}

\newcommand{\dist}[2][]{\ensuremath{\operatorname{dist}}\ifx\relax#1\relax\else\ensuremath{_{#1}}\fi\ensuremath{(#2)}}
\newcommand{\adist}[2][]{\ensuremath{\overline{\operatorname{dist}}}\ifx\relax#1\relax\else\ensuremath{_{#1}}\fi\ensuremath{(#2)}}
\renewcommand{\deg}[2][]{\ensuremath{\operatorname{deg}}\ifx\relax#1\relax\else\ensuremath{_{#1}}\fi\ensuremath{(#2)}}
\newcommand{\pred}[2][]{\ensuremath{\operatorname{N}^\text{in}}\ifx\relax#1\relax\else\ensuremath{_{#1}}\fi\ensuremath{(#2)}}
\newcommand{\indeg}[2][]{\ensuremath{\operatorname{indeg}}\ifx\relax#1\relax\else\ensuremath{_{#1}}\fi\ensuremath{(#2)}}
\newcommand{\suc}[2][]{\ensuremath{\operatorname{N}^\text{out}}\ifx\relax#1\relax\else\ensuremath{_{#1}}\fi\ensuremath{(#2)}}
\newcommand{\outdeg}[2][]{\ensuremath{\operatorname{outdeg}}\ifx\relax#1\relax\else\ensuremath{_{#1}}\fi\ensuremath{(#2)}}

\newcommand{\nh}[2][]{\ensuremath{\operatorname{N}}\ifx\relax#1\relax\else\ensuremath{_{#1}}\fi\ensuremath{(#2)}}
\newcommand{\inc}[2][]{\ensuremath{\operatorname{inc}}\ifx\relax#1\relax\else\ensuremath{_{#1}}\fi\ensuremath{(#2)}}
\newcommand{\LCA}[2][]{\ensuremath{\operatorname{LCA}}\ifx\relax#1\relax\else\ensuremath{_{#1}}\fi\ensuremath{(#2)}}

\newcommand{\LW}[3][]{\ensuremath{\operatorname{C}}\ifx\relax#1\relax\else\ensuremath{_{#1}}\fi^{#2}\ensuremath{(#3)}}
\newcommand{\CW}[2][]{\ensuremath{\operatorname{C}}\ifx\relax#1\relax\else\ensuremath{_{#1}}\fi\ensuremath{(#2)}}
\newcommand{\cw}[2][]{\ensuremath{\operatorname{cw}}\ifx\relax#1\relax\else\ensuremath{_{#1}}\fi\ensuremath{(#2)}}
\newcommand{\RW}[2][]{\ensuremath{\operatorname{R}}\ifx\relax#1\relax\else\ensuremath{_{#1}}\fi\ensuremath{(#2)}}
\newcommand{\rw}[2][]{\ensuremath{\operatorname{rw}}\ifx\relax#1\relax\else\ensuremath{_{#1}}\fi\ensuremath{(#2)}}
\newcommand{\cost}[2][]{\ensuremath{\operatorname{cost}}\ifx\relax#1\relax\else\ensuremath{_{#1}}\fi\ensuremath{(#2)}}
\newcommand{\Rcost}[2][]{\ensuremath{\operatorname{cost}^R}\ifx\relax#1\relax\else\ensuremath{_{#1}}\fi\ensuremath{(#2)}}
\newcommand{\HW}[2][]{\ensuremath{\operatorname{H}}\ifx\relax#1\relax\else\ensuremath{_{#1}}\fi\ensuremath{(#2)}}
\newcommand{\hw}[2][]{\ensuremath{\operatorname{hrl}}\ifx\relax#1\relax\else\ensuremath{_{#1}}\fi\ensuremath{(#2)}}
\newcommand{\TW}[2][]{\ensuremath{\operatorname{T}}\ifx\relax#1\relax\else\ensuremath{_{#1}}\fi\ensuremath{(#2)}}
\newcommand{\gw}[2][]{\ensuremath{\operatorname{\gamma{}w}}\ifx\relax#1\relax\else\ensuremath{_{#1}}\fi\ensuremath{(#2)}}

\newcommand{\N}{\ensuremath{\mathbb{N}}}

\newcommand{\PROB}[2]{\expandafter\newcommand\csname #1\endcsname{\textsc{#2}\xspace}}
\PROB{RA}{Register Allocation}
\PROB{HP}{Hamiltonian Path}
\PROB{HC}{Hamiltonian Cycle}
\PROB{DHP}{Directed Hamiltonian Path}
\PROB{DHC}{Directed Hamiltonian Cycle}
\PROB{ST}{Steiner Tree}
\PROB{TSAT}{3-Satisfyability}
\PROB{WTSAT}{Weighted 2-Satisfyability}
\PROB{MWTS}{Monotone Weighted 2-Satisfyability}
\PROB{VC}{Vertex Cover}
\PROB{HS}{Hitting Set}
\PROB{IS}{Independent Set}
\PROB{OCT}{Odd Cycle Transversal}
\PROB{TSP}{Traveling Salesman}
\PROB{HRL}{Horizontal Red Line Traversal}
\PROB{NAESAT}{Not-All-Equal 3-SAT}
\PROB{TC}{Tree Containment}
\PROB{CC}{Cluster Containment}
\PROB{STC}{Soft Tree Containment}
\PROB{TUIS}{2-Union Independent Set}
\PROB{LCS}{Longest Common Subsequence}
\PROB{SCSe}{Shortest Common Supersequence}
\PROB{SCSt}{Shortest Common Superstring}
\PROB{MCSP}{Minimum Common String Partition}
\PROB{EFSP}{Equality-Free String Partition}
\PROB{CS}{Closest String}
\PROB{CSu}{Closest Substring}

\PROB{SumCSP}{Sum CSP}

\newcommand{\smashsum}[2][lr]{\ensuremath{\smashoperator[#1]{\sum_{#2}}}}

\newcommand\nil[1][.7em]{%
	\makebox[#1]{%
		\kern.07em
		\vrule height.6ex
		\hrulefill
		\vrule height.6ex
		\kern.07em
	}
}

\ifx\backref\undefined
\else
\renewcommand*{\backref}[1]{}
\renewcommand*{\backrefalt}[4]{%
\ifcase #1 %
(Not cited.)%
\or
(Cited on page~#2.)%
\else
(Cited on pages~#2.)%
\fi}

\fi

\makeatletter
\def\NAT@spacechar{~}
\makeatother

\newcommand{\mybox}[2]{
  \begin{tikzpicture}
    \node[minimum width=\linewidth-0.4pt, draw, rounded corners, text width=\linewidth-12pt] (a){#2};
    \node[fill=white, xshift=1em, anchor=west] at (a.north west) {#1};
  \end{tikzpicture}
}

\newcommand{\myboxprobdef}[5]{
  \label{#5}
  \mybox{%
    \ifthenelse{\equal{#3}{}}{}{{#3}\ifthenelse{\equal{#4}{}}{}{ ({#4})}}
  }{%
    \begin{compactdesc}
      \item [Input:] {#1}
      \item [Question:] {#2}
    \end{compactdesc}
  }
}

\newcommand{\probdef}[5]{
\hbox{\vbox{
\begin{quote}
  \label{#5}
  \ifthenelse{\equal{#3}{}}{}{{#3}\ifthenelse{\equal{#4}{}}{}{ ({#4})}}
  \vspace{-1ex}
  \begin{compactdesc}
    \item [Input:] {#1}
    \item [Question:] {#2}
  \end{compactdesc}
\end{quote}
}}
}

\newcommand{\taskprobdef}[5]{
\hbox{\vbox{
\begin{quote}
  \label{#5}
  \ifthenelse{\equal{#3}{}}{}{{#3}\ifthenelse{\equal{#4}{}}{}{ ({#4})}}
  \begin{compactdesc}
    \item [Input:] {#1}
    \item [Task:] {#2}
  \end{compactdesc}
\end{quote}
}}
}

\newcommand{\paraproblem}[6]{
\hbox{\vbox{
\begin{quote}
  \label{#6}
  \ifthenelse{\equal{#4}{}}{}{{#4}\ifthenelse{\equal{#5}{}}{}{ ({#5})}}
  \begin{compactdesc}
    \item [Input:] {#1}
    \item [Question:] {#2}
    \item [Parameter:] {#3}
  \end{compactdesc}
\end{quote}
}}
}

\hypersetup
{
bookmarksopen=false,
pdftoolbar=false, 
pdfmenubar=true, 
pdfhighlight=/O, 
colorlinks=true, 
pdfpagemode=UseNone, 
pdfpagelayout=SinglePage, 
pdffitwindow=true, 
linkcolor=linkcol, 
citecolor=citecol, 
urlcolor=linkcol 
}

\makeatletter
\renewcommand\bibsection%
{
  \section*{\refname
    \@mkboth{\MakeUppercase{\refname}}{\MakeUppercase{\refname}}}
}
\makeatother

\begin{document}

\maketitle

\begin{abstract}
  The Neighbor Joining Algorithm is among the most fundamental algorithmic results in computational biology.
  However, its definition and correctness proof are not straightforward.
  In particular, ``the question ‘‘what does the NJ method seek to do?’’ has until recently proved somewhat elusive'' [Gascuel \& Steel, 2006].
  While a rigorous mathematical analysis is now available,
  it is still considered somewhat hard to follow and its proof tedious at best.
  
  In this work, we present an alternative interpretation of the goal of the Neighbor Joining algorithm
  by proving that it chooses to merge the two taxa~$u$ and $v$ that maximize the ``leaf-status'',
  that is, the sum of distances of all leaves to the unique $u$-$v$-path.
\end{abstract}

\section{Introduction and Preliminaries}\label{sec:prelim}

\paragraph{Neighbor Joining.}

Given $n$ taxa, as well as pairwise distances%
\footnote{The distances between taxa can be obtained in a variety of ways, for example by comparing their genomes.}
$d(u,v)$, the Neighbor Joining algorithm finds an undirected edge-weighted tree~$T$ whose leaves correspond to the $n$ taxa
and such that each two taxa $u$ and $v$ have distance $d(u,v)$ in $T$, provided such a tree exists (the distances are called ``additive'' in this case).
In this work, we call any such $T$ a \emph{representation} of the distances~$d$.
To find a representation for the input distances,
the Neighbor Joining algorithm finds a pair of taxa that are guaranteed to form a cherry in some representation~$T$ of $d$.
It then ``merges'' these two taxa into a new taxon corresponding to the mid-point of the cherry in $T$.
Finally, the distances $d$ are updated and a tree representing these new distances is found recursively.
See \cite{SOW+96,Fel04} for a more thorough explanation and examples for the Neighbor Joining algorithm.

A drawback of the algorithm is that finding a pair of taxa that form a cherry in a representation is a very opaque process,
to the point where the technique has been called ``obscure'' or even ``black magic''~\cite{Pev13}.
Efforts to render this important step more comprehensible and intuitive have been conducted~\cite{VD91,CHP93, Att99, SOW+96}
(see the excellent collection by \citet{Bryant05}, indicating why each of them is not fully satisfying).
\citet{GS06} concluded that ``NJ greedily optimizes a natural tree length estimate''
by selecting ``at each step as neighbors that pair of current taxa, which most decreases the whole tree length,
as computed using the generalized Pauplin formula''.
While this is a solid mathematical characterization, the ``whole tree length'' is hardly intuitive and
proving that a pair that most decreases it should form a cherry in a representation of $d$ is lengthy and tedious.
Neighbor Joining will produce a tree even for non-additive distances and
it has been claimed that the algorithm does not explicitly optimize any criterion in this case~\cite{SOW+96, Gas00, Fel04}.

In this work, we follow a similar path as \citet{VD91},
giving a simple and intuitive reformulation of the pair-selection criterion that
allows easy verification of correctness for additive distances and hopefully helps teach this important algorithm in class.
Indeed, it turns out that taxa $u$ and $v$ are selected to be merged if their hypothetical parent
is furthest from the center of the representing tree, more precisely, it has maximum ``leaf-status'' (see below).

Let us have a closer look at the process of selecting taxa to be merged by the Neighbor Joining algorithm.
From the distances $d(u,v)$, new quantities $q(u,v)$ are computed as follows.
\begin{equation}\label{eq:NJ1}
  q(u,v) := (n-2)\cdot d(u,v) - \sum_x d(u,x) - \sum_x d(v,x)
\end{equation}
Then, the pair~$(u,v)$ minimizing $q(u,v)$ is selected to be merged.
Correctness is established by \citet{SK88},
who showed that this formulation is equivalent to the one used by \citet{SN87}.
Indeed, leaves $u$ and $v$ minimizing $q(u,v)$ form a cherry in the sought tree representing the distances~$d$ \cite{Att99,EL09}.

\begin{theorem}[``Neighbor Joining Theorem'']\label{thm:NJ1}
  Let $D$ be a distance matrix representable by a tree and
  let $u$ and $v$ be such that $q(u,v)$, as defined in \eqref{eq:NJ1}, is minimum.
  Then, there is a tree~$T$ representing $D$ in which $(u,v)$ is a cherry.
\end{theorem}

\paragraph{Trees, Status and Leaf-Status.}
Let $T$ be an undirected tree.
We denote its nodes by $V(T)$, its edges by $E(T)$ and its leaves (nodes with degree one) by $L(T)$.
Two leaves $u,v\in L(T)$ are said to \emph{form a cherry} in $T$ if they are adjacent to the same node~$x$
and $x$ is called the \emph{mid-point} of the cherry.
In this work, we will consider trees~$T$ that do not contain nodes of degree two.
Let the edges of $T$ be weighted by a function $\w:E(T)\to\N^+$.
The \emph{distance}~$d_T(u,v)$ between nodes $u,v\in V(T)$ is the weight of the unique $u$-$v$-path in $T$.
For a $u$-$v$-path~$p$ in $T$, we define the \emph{distance} of a node~$x$ to $p$ as
the smallest distance of~$x$ to any node on $p$, that is, 
$d_T(x,p):=\min_{u\in V(p)}d(u,x)$.

\begin{observation}\label{obs:path dist}
  Let $u,v,x\in V(T)$. The distance of $x$ to the $u$-$v$-path $p$ is $d_T(x,p) = \nicefrac{1}{2}\left( d(x,u) + d(x,v) - d(u,v) \right)$.
\end{observation}

\noindent
For any edge~$uv$ of $T$ let $L^{uv}_u$ denote the set of leaves~$x$ that are closer to $u$ than to $v$ in $T$,
that is $L^{uv}_u:=\{x\in L(T)\mid d_T(x,u)<d_T(x,v)\}$.
The \emph{status} of $u\in V(T)$ in $T$ is the sum of all distances from $u$, that is, $s_T(u):=\sum_{x\in V(T)}d(u,x)$.
While the status has been researched in the past~\cite{EJS76,LS09,SL11,LTS+12,QZ20}, we require a slight variation of the concept:
The \emph{leaf-status} of $u\in V(T)$ in $T$ is the sum of distances of $u$ to all leaves, that is, $\ell_T(u):=\sum_{x\in L(T)}d(u,x)$.
The leaf-status of a path~$p$ of $T$ is the sum of distances of all leaves to $p$, that is, $\ell_T(p):=\sum_{x\in L(T)} d(x,p)$.
With \cref{obs:path dist}, we observe that the leaf-status of a path can be formulated in terms of the leaf-statuses of its endpoints.
Note how closely this formulation resembles \Cref{eq:NJ1}.

\begin{observation}\label{obs:path status}
  Let $u,v\in V(T)$ and
  let $p$ be a $u$-$v$-path in $T$.
  Then, $\ell_T(p)=\nicefrac{1}{2}(\ell_T(u) + \ell_T(v) - |L(T)|\cdot d(u,v))$.
\end{observation}

\noindent
In the following, we will omit the subscript if $T$ is clear from the context.

\paragraph{Working with the Leaf-Status.}
As a warm-up, we prove a version%
\footnote{The original property is proved for the status in unweighted graphs, while ours is for the leaf-status in edge-weighted trees.}
of \cite[Property~2.2]{EJS76} for the leaf-status.

\begin{lemma}\label{lem:close}
  Let $uv\in E(T)$.
  Then, $\ell(u) - \ell(v) = \w(uv)\left(|L^{uv}_v|-|L^{uv}_u|\right)$.
\end{lemma}
\begin{proof}\phantom{.}\\[-6ex]
  \begin{align*}
    \ell(u) & = \sum_{a\in L^{uv}_u}d(a,u) + \sum_{b\in L^{uv}_v}d(b,u)\\
            & = \sum_{a\in L^{uv}_u}(d(a,u)+\w(uv) - \w(uv)) + \sum_{b\in L^{uv}_v}(d(b,u)-\w(uv) + \w(uv))\\
            & = \sum_{a\in L^{uv}_u}(d(a,v) - \w(uv)) + \sum_{b\in L^{uv}_v}(d(b,v) + \w(uv))
            \;=\; \ell(v) + \w(uv)(|L^{uv}_v|-|L^{uv}_u|)\qedhere
  \end{align*}
\end{proof}

\begin{corollary}\label{cor:close}
  Let $T$ be a tree with at least three leaves and
  let $x\in V(T)$ minimize $\ell(x)$.
  Then, $x$ is not a leaf.
\end{corollary}

\noindent
A central theorem concerning the status is that, for any path $(v_0,v_1,\ldots,v_k)$ in $T$ originating in a node $v_0$ of minimum status in $T$,
we have $s(v_0)\leq s(v_1)<s(v_2)<\ldots<s(v_k)$~\cite[Theorem~3.3]{EJS76}.
Using \cref{lem:close}, we show a similar version for the leaf-status of $T$ even when edges are positively weighted.

\begin{lemma}\label{lem:path stat}
  Let $T$ be a tree with edge-weights $\w:E(T)\to\N^+$ that is free of degree-2 nodes.
  Let $(v_0,v_1,\ldots,v_k)$ be a path in $T$ and
  let $v_0$ have minimum leaf-status in $T$.
  Then, $\ell(v_0)\leq \ell(v_1) < \ell(v_2) < \ldots < \ell(v_t)$.
\end{lemma}
\begin{proof}
  For all $i$, we abbreviate $L^\leftarrow_i:=L^{v_iv_{i+1}}_{v_i}$,
  that is, $L^\leftarrow_i$ is the set of leaves that are closer to $v_i$ than to $v_{i+1}$,
  and $L^\rightarrow_i:=L^{v_iv_{i+1}}_{v_{i+1}}$.
  Note that, as $T$ is a tree without degree-2 nodes,
  we have $L^\leftarrow_i\subset L^\leftarrow_{i+1}$ and  $L^\rightarrow_i\supset L^\rightarrow_{i+1}$,
  implying $|L^\leftarrow_i|<|L^\leftarrow_{i+1}|$ and  $|L^\rightarrow_i|>|L^\rightarrow_{i+1}|$.
  Further, for all~$i$, as $\w(v_iv_{i+1})>0$, we have
  \begin{align}\label{eq:ell}
    \ell(v_i) \leq \ell(v_{i+1})
    \;\stackrel{\cref{lem:close}}{\iff}\;
    |L^\leftarrow_i|
    \geq
    |L^\rightarrow_i|
    &&
    \text{and}
    &&
    \ell(v_i) = \ell(v_{i+1})
    \;\stackrel{\cref{lem:close}}{\iff}\;
    |L^\leftarrow_i|
    =
    |L^\rightarrow_i|
  \end{align}

  \noindent
  We prove the lemma by induction on $i$.
  For $i=0$, $\ell(v_i)\leq\ell(v_{i+1})$ by definition of $v_0$.
  Otherwise, $\ell(v_{i-1})\leq\ell(v_i)$ by induction hypothesis, implying
  $|L^\leftarrow_i| > |L^\leftarrow_{i-1}| \geq |L^\rightarrow_{i-1}| > |L^\rightarrow_i|$ by \eqref{eq:ell}.
  But then, $\ell(v_i) < \ell(v_{i+1})$ by \eqref{eq:ell}.
\end{proof}

\noindent
\cref{lem:path stat} implies that, in analogy to the status, the nodes~$v$ minimizing the leaf-status are ``in the center'' of $T$
and that $\ell(v)$ is a measure of ``elongation from the center'' of $T$.

\section{Results}\label{sec:res}

Instead of minimizing the ``ominous'' values $q(u,v)$,
we propose to \emph{maximize} the leaf-status of the $u$-$v$-path~$p$ in the target tree~$T$, augmented by the length $d(u,v)$ of $p$:
\begin{equation}\label{eq:NJ2}
  z(u,v) := d(u,v) + \ell(p) = d(u,v) + \smashsum[r]{x\in L(T)}d(x,p) \stackrel{\cref{obs:path dist}}{=} d(u,v) + \nicefrac{1}{2} \sum_{x\in L(T)}\left(d(u,x) + d(v,x) - d(u,v)\right)
\end{equation}
Indeed, $q(u,v) = -2z(u,v)$ can be easily verified for all taxa $u$ and $v$ and, thus, maximizing $z(u,v)$ is equivalent to minimizing $q(u,v)$.
The following key property connects $z(u,v)$ to the leaf-status in a representation $T$ for $d$.

\begin{lemma}\label{lem:zij extremity}
  Let $T$ be a tree with positive edge-weights without degree-2 nodes and
  let $u,v\in L(T)$.
  Let $w$ be an inner node of the unique $u$-$v$-path~$p$ in $T$.
  Then, $z(u,v)\leq\ell(w)$ and equality holds if and only if $p=(u,w,v)$.
\end{lemma}
\begin{proof}
  By definition, for all $x\in L(T)$, we have $d(x,p)\leq d(x,w)$ and,
  as $T$ is free of degree-2 nodes and its edge-weights are positive,
  equality holds if and only if the $x$-$w$-path avoids all edges of $p$.
  Since $d(u,v)=d(u,w)+d(w,v)$, we conclude
  \begin{align*}
    z(u,v)
    \stackrel{\eqref{eq:NJ2}}{=} d(u,v) + \smashsum{x\in L(T)} d(x,p)
    = d(u,v) + \smashsum{x\in L(T)\setminus\{u,v\}} d(x,p)
    \leq d(u,w) + d(v,w) + \smashsum{x\in L(T)\setminus\{u,v\}} d(x,w)
    = \smashsum{x\in L(T)} d(x,w) = \ell(w),
  \end{align*}
  and equality holds if and only if all paths from leaves~$x\in L(T)\setminus\{u,v\}$ to $w$ are edge-disjoint to $p$ or,
  equivalently, all edges in $p$ are incident with either $u$ or $v$, that is, $p=(u,w,v)$.
\end{proof}

As the main contribution, the following theorem establishes the connection of $z(u,v)$ to the leaf-status,
thereby giving a simple, intuitive explanation of the correctness of the Neighbor Joining algorithm:
it merges leaves~$u$ and $v$ that, in each representation~$T$ for $d$, form a cherry whose mid-point has maximum leaf-status and, by \cref{lem:path stat},
maximum elongation from the center of $T$.

\begin{theorem}\label{thm:NJ status}
  Let $d$ be distances for at least three taxa, representable by a tree~$T$ and
  let $u,v\in L(T)$ such that $z(u,v)$ is maximum.
  Then, $u$ and $v$ form a cherry in $T$ and $\ell(w)=z(u,v)$ for its mid-point~$w$.
\end{theorem}
\begin{proof}
  Assume towards a contradiction that $u$ and $v$ do not form a cherry in $T$.
  Let $c$ be a node of minimum leaf-status in $T$ and
  let $p$ be the unique $c$-$u$-path in $T$.
  By \cref{cor:close}, $c$ is not a leaf and the predecessor~$x$ of $u$ in $p$ is not a leaf either.
  
  By \cref{lem:zij extremity} and the assumption that $u$ and $v$ are not a cherry, $z(u,v) < \ell(x)$.
  Since $x$ does not have degree two in $T$, we know that $x$ has distinct neighbors $s$ and $t$, none of which is $u$ or $v$.
  Without loss of generality, let $\ell(s)\geq\ell(t)$.
  If $s$ is a leaf in $T$, then $z(u,s) = \ell(x) > z(u,v)$ by \cref{lem:zij extremity}, contradicting maximality of $z(u,v)$.
  Thus, suppose that $s$ is not a leaf in $T$.
  Clearly, \cref{lem:path stat} forbids $\ell(s),\ell(t) \leq \ell(x)$ and, thus, $\ell(s) > \ell(x)$.
  Of the two connected components of the result of removing the edge~$xs$ from $T$, let $T_s$ be the one containing~$s$.
  Note that
  (a)~$u,c\notin V(T_s)$ and (b)~$\ell(y) > \ell(x)$ for all $y\in V(T_s)$ by \cref{lem:path stat}.
  Since $s$ is not a leaf in $T$, there is a cherry in $T_s$ that is also a cherry in $T$ and, by \cref{lem:zij extremity},
  this cherry contradicts maximality of $z(u,v)$.
  Thus, $u$ and $v$ form a cherry in $T$ and, by \cref{lem:zij extremity}, $\ell(x)=z(u,v)$.
\end{proof}

\section{Conclusion and Discussion}
We presented a new, intuitive point of view onto the well-known Neighbor Joining algorithm.
We show that finding the minimum in a matrix of ``obscure'' entries can be interpreted
as finding the internal node~$w$ maximizing the ``leaf-status'', that is,
the sum of distances to $w$ in the target tree~$T$.
While the proof of the Neighbor Joining Theorem is usually skipped in class due to its
complexity, our simplified explanation and proof can be taught to undergraduate students with reasonable effort.

The proof that we present here relies on the edge-weights in the target tree $T$ being positive.
Indeed, upon closer inspection, the proof only requires the edge-weights of all \emph{internal} edges of $T$ to be positive.
Thus, our results are applicable to a modification of the Neighbor Joining algorithm that constructs
the result of contracting all weight-zero edges between internal nodes of $T$.
One can realize that, in presence of zero-weight internal edges in $T$, the Neighbor Joining algorithm
returns an \emph{arbitrary} binary tree~$T'$ such that contracting all zero-weight internal edges in $T$ and in $T'$
results in the same tree.

\let\oldthebibliography=\thebibliography
\let\endoldthebibliography=\endthebibliography
\renewenvironment{thebibliography}[1]{%
  \begin{oldthebibliography}{#1}%
    \setlength{\parskip}{0ex}%
    \setlength{\itemsep}{0ex}%
    \footnotesize
}{
    \end{oldthebibliography}%
}
\bibliographystyle{abbrvnat}

\bibliography{nj}

\end{document}